%% file: dec51-no_line_numbers.tex
\documentclass[11pt]{article}

\usepackage{amsmath, amssymb, latexsym, amsthm}
\usepackage{mathabx}
\usepackage{bm}
\usepackage{wasysym}
\usepackage{fullpage}
\usepackage{color}
\usepackage{tikz}
\usepackage{graphicx}
\usepackage{stmaryrd,amsbsy,enumerate}
\usepackage{hyperref}

\usepackage{mathscinet}
\usetikzlibrary{calc}
\usepackage{cancel}
\usepackage{xspace}
\usepackage[mathlines]{lineno}
\usepackage{url}
\usepackage{doi}
\usepackage{longtable}
\usepackage[obeyFinal,colorinlistoftodos,textsize=footnotesize]{todonotes}
\usepackage{extarrows}
\usepackage[affil-it,auth-sc]{authblk}
\usepackage{enumitem}
\usepackage{mathtools,amsmath}
\usepackage{array}
\usepackage[margin=10pt,font=small,labelfont=sc,labelsep=period]{caption}
\usepackage{soul}
\usepackage{relsize}

\def\real{\mathbb{R}}
\def\H#1#2#3{{{#1}[#2;#3]}}

\def\ww{{\cal W}}
\def\vv{\mathcal V}
\def\zz{\mathcal Z}
\def\ora#1{{\overrightarrow{#1}}}
\def\set#1{{\text{\rm set}(#1)}}
\def\kl{{\text{\rm cl}}}
\def\kint#1#2{{\text{\rm sh}(#1,#2)}}
\def\Ind{{\text{\rm Ind}}}

\def\perm#1{{\llbracket{#1}\rrbracket}}
\def\cycper#1{{\llparenthesis{\,}{#1}{\,}\rrparenthesis}}

\newcommand\cca[1]{{%
  \ooalign{\raisebox{-.2ex}{\larger[3]$\circlearrowright$}\cr
    \hidewidth$#1$\hidewidth}}}

\newcommand*{\Scale}[2][4]{\scalebox{#1}{\ensuremath{#2}}}%

\newtheorem{theorem}{Theorem} 
\newtheorem{theorem*}{Theorem}

\newtheorem{corollary}[theorem]{Corollary}
\newtheorem{lemma}[theorem]{Lemma}

\newtheorem{claim}[theorem]{Claim}

\newtheorem{observation}[theorem]{Observation}

\theoremstyle{definition}

\newtheorem{definition}{Definition}

\begin{document}


\title{Deciding monotonicity of simple drawings of the complete graph}

\author[1]{Oswin Aichholzer\thanks{\textit{E-mail:} \texttt{oswin.aichholzer@tugraz.at.}}}

\author[2]{Thomas Hackl\thanks{This work was conducted while this author was a postdoctoral researcher at the Institute of Software Technology, Graz University of Technology. He is currently employed in the private sector.}}

\author[2]{Alexander Pilz\textsuperscript{$\dag$}}

\author[3]{Gelasio Salazar\thanks{\textit{E-mail:} \texttt{gelasio.salazar@uaslp.mx.}}}
\author[1]{Birgit Vogtenhuber\thanks{\textit{E-mail:} \texttt{birgit.vogtenhuber@tugraz.at.}}}

\affil[1]{Institute of Algorithms and Theory, Graz University of Technology, Austria}
\affil[2]{Institute of Software Technology, Graz University of Technology, Austria}
\affil[3]{Instituto de F\'\i sica, Universidad Aut\'onoma de San Luis Potos\'{\i}, Mexico}

\maketitle

\begin{abstract}
A drawing of a graph is {\em $x$-monotone} if every vertical line intersects each edge of the graph at most once. We present an $O(n^5)$ time algorithm for deciding whether a simple drawing of the complete graph $K_n$ is weakly isomorphic to an $x$-monotone drawing. We note that this algorithm can also decide whether a drawing of $K_n$ is strongly isomorphic to an $x$-monotone drawing.
\end{abstract}

\section{Introduction}\label{sec:intro}

In a {\em drawing} of a graph on some surface, vertices are represented by distinct points and each edge is represented by a Jordan arc whose endpoints are the endvertices of the edge. No edge passes through a vertex, and edges intersect each other in a finite number of points. For simplicity usually it is also assumed that no three edges meet at a common interior point.

Throughout this paper we work with simple drawings of the complete graph $K_n$. We recall that in a {\em simple} drawing (also known as a {\em good} drawing) in addition to the previous properties no two edges share more than one point (either a common endvertex or a proper crossing), and no edge crosses itself. An important motivation for investigating simple drawings is that every crossing-minimal drawing of a graph is simple~\cite{schaefer-2018-cng}. The study of simple drawings and their substructures has attracted significant interest in a variety of contexts~\cite{shootingstars,twisted,bergold2025plane,fulekvargas,kynclimproved,pst,tangled-thrackle,schaeferdetour,sukzeng}. We emphasize that throughout this work all drawings under consideration are implicitly assumed to be simple and unless otherwise stated are hosted in the plane $\real^2$.

A drawing is $x$-\textit{monotone} if every vertical line intersects each edge at most once. See Figure~\ref{fig:150} for an illustration. Our main results involve the existence of polynomial time algorithms to test whether a given drawing of $K_n$ in the plane is (weakly or strongly) isomorphic to an $x$-monotone drawing.

We also recall that two drawings $D,D'$ of the same graph are {\em weakly isomorphic} if there is an incidence-preserving bijection between the drawings such that two edges cross in $D$ if and only if their images in $D'$ cross. Now $D$ and $D'$ are {\em strongly isomorphic} if they induce homeomorphic cell decompositions of the sphere. That is, they are strongly isomorphic if $D'$ can be obtained from $D$ by performing an inverse stereographic projection to the sphere, followed by a self-homeomorphism of the sphere, and finally followed by a stereographic projection back to the plane.

There seem to be very few algorithmic results related to $x$-monotone drawings reported in the literature. Fulek, Pelsmajer, Schaefer, and {\v{S}}tefankovi{\v{c}}~\cite{fpss} gave an $O(n^2)$ time algorithm that tests whether a graph with given $x$-coordinates assigned to the vertices has an $x$-monotone embedding (respecting the given $x$-coordinates). Recently, Kyn\v{c}l and Soukup established an NP-hardness result on the related notion of cylindrical monotonicity~\cite{extendingm}.

Our main result is the following.

\begin{theorem}\label{thm:mon}
There is an $O(n^5)$ time algorithm that decides if a given drawing of $K_n$ is weakly isomorphic to an $x$-monotone drawing.
\end{theorem}

We remark that in Theorem~\ref{thm:mon} and throughout this paper we do not assume that the input drawings of $K_n$ are given in any particular format. We just assume that drawings are given in some ``reasonable'' data structure from which we can obtain the cell structure and the extended rotation system (see~\cite{complexity_kyncl,enukyncl}) of the drawing in $O(n^5)$ time. Since our ultimate goal is to prove Theorem~\ref{thm:mon} we may as well assume that we are given from the onset both the cell structure and the extended rotation system of any input drawing.

We note that Theorem~\ref{thm:mon} easily implies an analogous result under strong isomorphism. 

To see this we first recall Gioan's theorem~\cite{arroyogioan,gioan,schaeferdetour}, which states that if we have two drawings in the sphere with the same rotation system then one can be obtained from the other by a sequence of triangle mutations. We also recall that two drawings are weakly isomorphic if and only if their rotation systems are equivalent (that is, perhaps after some relabelling of the vertices, they are either identical or the reverse of one another). We finally note that it is easy to see that in the plane any triangle mutation on an $x$-monotone drawing can be performed while keeping its $x$-monotonicity. Combining these three facts it is straightforward to see that a drawing of $K_n$ is weakly isomorphic to an $x$-monotone drawing if and only if it is strongly isomorphic to an $x$-monotone drawing. Therefore Theorem~\ref{thm:mon} implies the following.

\begin{corollary}\label{cor:mon2}
There is an $O(n^5)$ time algorithm that decides if a given drawing of $K_n$ is strongly isomorphic to an $x$-monotone drawing.
\end{corollary}

A preliminary version of this work has been presented at the XVI Spanish Meeting on Computational Geometry~\cite{ahpsv}.

\section{Proof of Theorem~\ref{thm:mon}}\label{sec:proofmon}

For brevity, for the rest of this paper we say that a drawing is {\em monotone} if it is weakly isomorphic to an $x$-monotone drawing. Under this terminology Theorem~\ref{thm:mon} reads as follows.

\vglue 0.3 cm
\noindent{\bf Theorem~\ref{thm:mon}. } {\em There is an $O(n^5)$ time algorithm that decides if a given drawing of $K_n$ is monotone.}
\vglue 0.3 cm

In this section we reduce Theorem~\ref{thm:mon} to two lemmas. The first one (namely Lemma~\ref{lem:charac}) gives necessary and sufficient conditions for a drawing of $K_n$ to be monotone, whereas the second one (namely Lemma~\ref{lem:n5}) states that these properties can be tested in $O(n^5)$ time. We finish the section by arguing that indeed Lemmas~\ref{lem:charac} and~\ref{lem:n5} together easily imply Theorem~\ref{thm:mon}. The rest of the paper is devoted to the proofs of these lemmas.

In the proof of Theorem~\ref{thm:mon} we make extensive use of rotation systems. We recall that the {\em rotation at} a vertex in a drawing $D$ is the cyclic permutation that records the clockwise cyclic order in which the edges (each denoted by the label of their `other' endpoint) incident with the vertex leave the vertex. The {\em rotation system} of $D$ is the collection of the rotations at all the vertices. 

\vglue 0.3 cm
\noindent{\bf Remark. }{\sl Throughout this paper we assume that the vertices of $K_n$ are labelled with the integers in $[n]=\{1,\ldots,n\}$.}
\vglue 0.3 cm


\subsection{Wedges}

The characterization of monotone drawings given in Lemma~\ref{lem:charac} below relies on the notion of a wedge, an object that plays a paramount role in this work. For an illustration of this concept we refer the reader to Figure~\ref{fig:150} and its caption. In a nutshell, a wedge is a consecutive subpermutation of a vertex rotation.

To make this precise and lay out the corresponding notation, for the rest of this paper we use $\cca{i}$ to denote the rotation at vertex $i$ in a drawing $D$ of $K_n$. In principle this notation should include some reference to the drawing $D$, but this is unnecessary as at all times we assume that we are working with a fixed drawing. In particular, in the next definition we assume that we are working with a fixed drawing of $K_n$.


\begin{definition}[Wedges]
{\sl If $i,j$ are distinct vertices and $r$ is an integer, $1\le r \le n-1$, then we use $\H{i}{j}{r}$ to denote the (linear) consecutive subsequence of $\cca{i}$ of length $r$ that starts at $j$. We say that $\H{i}{j}{r}$ is an $i$-{\em wedge}, or equivalently that it is a {\em wedge with hub vertex (or simply hub) $i$}. The {\em initial} vertex of this $i$-wedge is the first vertex in $\H{i}{j}{r}$ (namely $j$) and its {\em final} vertex is the last vertex in $\H{i}{j}{r}$.}
\end{definition}
 
\begin{figure}[ht!]
\def\ta#1{{\Scale[1.6]{#1}}}
\centering
\scalebox{0.3}{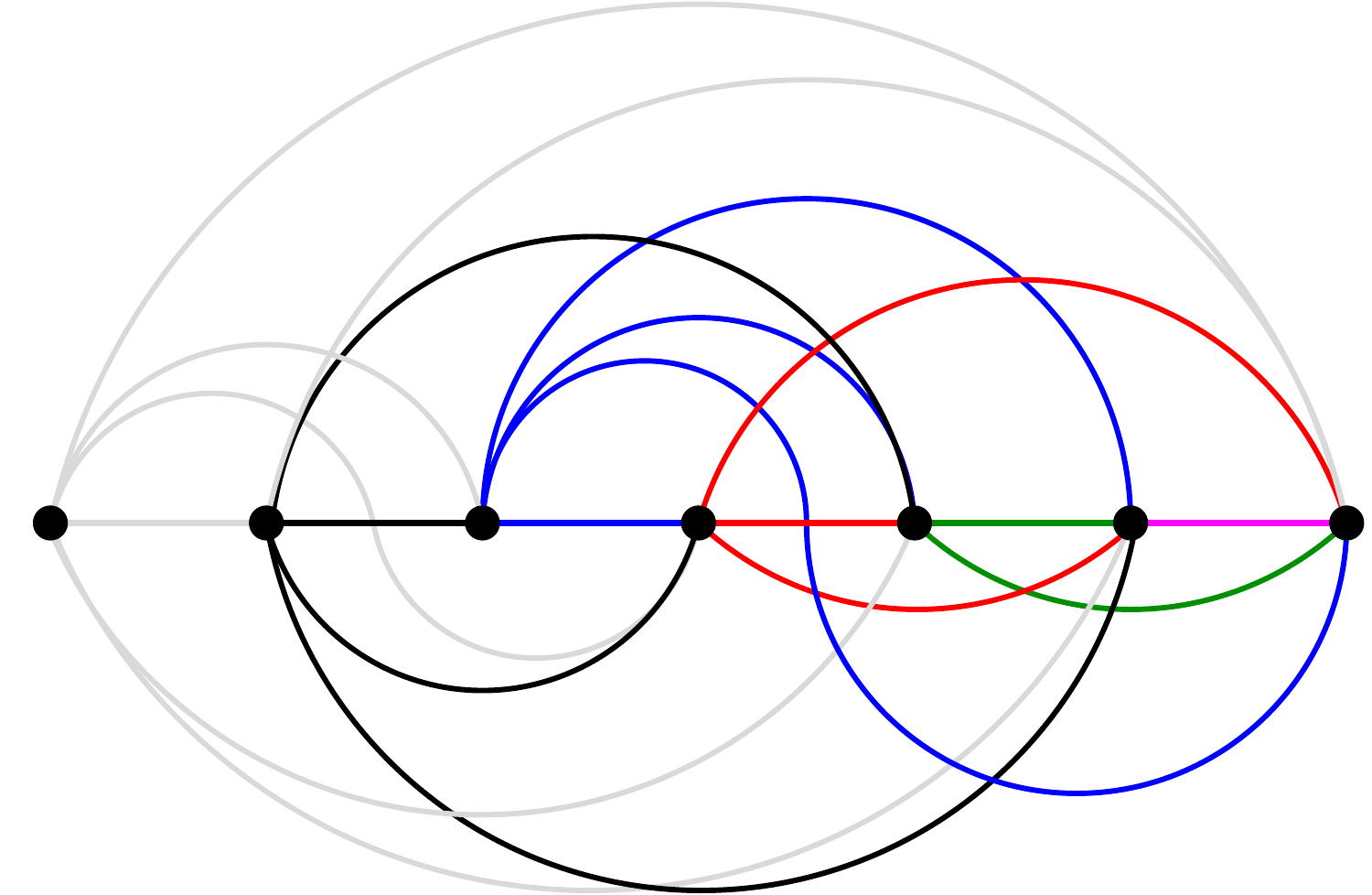}
\caption{An $x$-monotone drawing $D$ of $K_7$. In this drawing the $4$-wedge $\H{4}{6}{6}$ is $\perm{6,3,5,7,2,1}$, the $7$-wedge $\H{7}{6}{5}$ is $\perm{6,2,3,5,1}$, the $3$-wedge $\H{3}{1}{4}$ is $\perm{1,2,6,5}$, the $5$-wedge $\H{5}{6}{3}$ is $\perm{6,2,1}$, the $2$-wedge $\H{2}{2}{1}$ is $\perm{1,6}$, and the $1$-wedge $\H{1}{6}{1}$ is $\perm{6}$. The $x$-order of this drawing is $i_7=4,i_6=7, i_5=3, i_4=5, i_3=2, i_2=1, i_1=6$. Vertices $i_7$ and $i_1$ are on the boundary of a cell in $D$ (the unbounded cell), and so (M1) holds. Clearly for each $r=7,6,5,4,3$ the edge $i_r,i_{r-1}$ crosses no edge with both endvertices in $\{i_{r-2},\ldots,i_1\}$, and so (M3) holds. We finally note that if we write the wedges identified at the beginning of this caption using the labels $i_7,\ldots,i_1$ we obtain that $\{i_6,\ldots,i_1\}$ is the vertex set of an $i_7$-wedge, $\{i_5,\ldots,i_1\}$ is the vertex set of an $i_6$-wedge, $\{i_4,\ldots,i_1\}$ is the vertex set of an $i_5$-wedge, $\{i_3,\ldots,i_1\}$ is the vertex set of an $i_4$-wedge; $\{i_2,i_1\}$ is the vertex set of an $i_3$-wedge; and $\{i_1\}$ is the vertex set of an $i_2$-wedge. Thus (M2) holds.}
\label{fig:150}
\end{figure}

Note that while the rotation at a vertex is a cyclic permutation, a wedge is a linear permutation.

We emphasize that saying that the hub vertex of a wedge is $i$ is equivalent to saying that it is an $i$-wedge. We also note that the hub vertex of a wedge is not part of the wedge. Note that we allow the possibility that a wedge consists of a single vertex. Indeed, if $i,j$ are any two distinct vertices then $j$ is the only vertex in the $i$-wedge $\H{i}{j}{1}$.

For instance, in the $x$-monotone drawing in Figure~\ref{fig:150} the wedge $\H{3}{1}{4}$ is the (linear) permutation $\perm{1,2,6,5}$. We refer the reader to the caption of Figure~\ref{fig:150} for an identification of several wedges in that drawing.


\vglue 0.3 cm
\noindent{\bf Notation. }{\em When there is no need to specify explicitly its initial vertex and its length, we often refer to a wedge with a calligraphic upper case letter, typically $\ww$.}


\vglue 0.3 cm
\noindent{\bf Notation. }{\em Wedges are permutations of vertices, but many of our arguments (even the characterization of $x$-monotonicity in Lemma~\ref{lem:charac} below) involve the (unordered) set of vertices of a wedge. If $\ww=\H{i}{j}{r}$ is a wedge then we use $\set{\ww}$ (equivalently, $\set{\H{i}{j}{r}}$) to denote the set of its vertices.}
\vglue 0.3 cm


\subsection{Proof of Theorem~\ref{thm:mon}}

The characterization of monotone drawings given by Lemma~\ref{lem:charac} below is based upon three glaring properties satisfied by every $x$-monotone drawing. Suppose that in an $x$-monotone drawing the vertices are ordered $i_n, i_{n-1},\ldots,i_1$ from left to right, that is, in increasing order in their $x$-cooordinates. We say that $i_n,i_{n-1},\ldots,i_1$ is the {\em $x$-order} of the drawing. (Admittedly it seems more natural to choose the labelling using the increasing order $i_1,\ldots,i_n$, but our choice turns out to be a lot more convenient for some discussions and proofs). 

It is easy to see that if $D$ is an $x$-monotone drawing of $K_n$ with $x$-order $i_n,\ldots,i_1$, then (see for instance Figure~\ref{fig:150} for an illustration):

\begin{enumerate}

\item[(M1)] there is a cell $C$ of $D$ that has both $i_1$ and $i_n$ in its boundary;

\item[(M2)] for each $r=n,n-1,\ldots,{2}$, there is an $i_r$-wedge whose vertex set is $\{i_{r-1},\ldots,i_1\}$; and

\item[(M3)] for each $r=n,n-1,\ldots,3$, the edge $i_r i_{r-1}$ does not cross any edge with both endvertices in $\{i_{r-2},\ldots,i_1\}$.

\end{enumerate}

We note that (M2) (respectively, (M3)) trivially holds for $r=2$ (respectively, for $r=3$) for every labelling $i_n,\ldots,i_1$ in any (not necessarily $x$-monotone) drawing of $K_n$. Indeed, regardless of which vertices get labelled $i_2$ and $i_1$ there is a wedge with hub vertex $i_2$ whose only vertex is $i_1$, and so (M2) holds. Regarding (M3) it suffices to note that the edge $i_3 i_2$ cannot cross any edge with both endvertices in $\{i_1\}$, simply because $\{i_1\}$ contains only one vertex. 

We claim that these obviously necessary conditions guarantee that a drawing of $K_n$ is monotone:

\begin{lemma}\label{lem:charac}
A drawing $D$ of $K_n$ is monotone if and only if there is a sequence $(i_n,i_{n-1},\ldots,i_1)$ of its vertices that satisfies (M1)--(M3).
\end{lemma}

The proof of Lemma~\ref{lem:charac} is based on a characterization of monotonicity established by Balko, Fulek, and Kyn\v{c}l~\cite{kyncl_monotone}. In Section~\ref{sec:cha} we review their characterization, which is given in terms of shelling sequences. Also in that section we show that this characterization can be formulated in terms of the weaker notion of partial shelling sequences. With this in hand we prove Lemma~\ref{lem:charac} in Section~\ref{sec:proofcharac}.

The next lemma, whose proof is given in Section~\ref{sec:algo}, is the second main ingredient in the proof of Theorem~\ref{thm:mon}.

\begin{lemma}\label{lem:n5}
There is an $O(n^5)$ time algorithm that decides if a given drawing $D$ of $K_n$ has a sequence $i_n,\ldots,i_1$ of its vertices such that (M1)--(M3) hold in $D$. 
\end{lemma}

We finish this section by noting for completeness that Theorem~\ref{thm:mon} indeed follows easily from Lemmas~\ref{lem:charac} and~\ref{lem:n5}.

\begin{proof}[Proof of Theorem~\ref{thm:mon}]
Let $D$ be a drawing of $K_n$. By Lemma~\ref{lem:n5} in $O(n^5)$ time we may decide whether there is a sequence $i_n,\ldots,i_1$ of the vertices that satisfies (M1)--(M3) in $D$. In view of Lemma~\ref{lem:charac} this algorithm decides whether $D$ is monotone, and so we are done.
\end{proof}

\section{Characterizing Monotonicity Using Partial Shelling Sequences}\label{sec:cha}

Our aim in this section is to give a characterization of monotonicity (namely Lemma~\ref{lem:cha2} below) that draws upon a result by Balko, Fulek, and Kyn\v{c}l~\cite{kyncl_monotone}. As we shall see in the next section, this characterization is a key ingredient in the proof of Lemma~\ref{lem:charac}.

In Section~\ref{sub:bfk} we review the aforementioned result from~\cite{kyncl_monotone} (this is Lemma~\ref{lem:charoriginal} below), and use it to derive a slightly streamlined formulation better suited to our purposes (Lemma~\ref{lem:cha}). These statements involve the notion of a shelling sequence, which we review below. We finally put forward in Lemma~\ref{lem:cha2} of Section~\ref{sub:partial} our characterization, given in terms of the weaker notion of a {\em partial} shelling sequence. 

\subsection{A Characterization of Monotonicity Using Shelling Sequences}\label{sub:bfk}

In order to state the characterization given in~\cite{kyncl_monotone} we start by recalling that if a drawing $D$ is regarded as a subset of the plane, then a {\em cell} of $D$ is a connected component of $\real^2\setminus D$. (If $D$ has no edge crossings, that is, if it is an embedding, then one would usually refer to a cell as a face of the embedding). If $C$ is a cell of $D$ then we use $\partial{C}$ to denote its boundary. 

We finally recall the notion of a shelling sequence. Let $S=(i_n,\ldots,i_1)$ be a sequence of vertices in a drawing $D$ of $K_n$. For $n \ge s > r \ge 1$ we use $D(S,s,r)$ to denote the drawing obtained from $D$ by removing the vertices $i_n,\ldots,i_{s+1},i_{r-1},\ldots,i_1$ and their incident edges. 

\vglue 0.3 cm
\noindent{\bf Remark. }{\sl Throughout this paper we use the terms {\rm sequence} and {\rm permutation} (of vertices) interchangeably. In this section we exclusively use {\rm sequence}, as this is the customary terminology in the literature for the notions involved, such as the concept of a {\rm shelling sequence.}}
\vglue 0.3 cm

\begin{definition}[\'Abrego et al.~\cite{shellable_drawings}]\label{def:shellable}
{\sl Let $t\ge 1$ be an integer. A drawing $D$ of $K_n$ is {\em shellable} if there exists a sequence $S = (i_n , \ldots , i_1 )$ of vertices and a cell $C$ of $D$ with the following property. For all $n \ge s > r \ge 1$, the vertices $i_s$ and $i_r$ are on the boundary of the cell of $D(S, s, r)$ that contains $C$. The sequence $S$ is a {\em shelling sequence} (or simply a {\em shelling}) of $D$ witnessed by $C$.}
\end{definition}

We remark that the original definition of shellability is actually more general, as it involves sequences of $t$ vertices where $t$ may be strictly smaller than $n$ (thus one speaks of $t$-{\em shellable} sequences and $t$-{\em shellings}). Since we will not work in this context with sequences with fewer than $n$ vertices, it seems best to simply adopt the previous definition.

The characterization in~\cite{kyncl_monotone} focuses on when a drawing is weakly isomorphic to an $x$-monotone drawing with a given $x$-order. Formally, in~\cite{kyncl_monotone} a sequence $S=(i_n,\ldots,i_1)$ of vertices is said to be an $x$-{\em monotone} sequence of a drawing $D$ if $i_1$ and $i_n$ are incident with the unbounded cell of $D$ and $D$ is weakly isomorphic to a drawing with $x$-order $i_n,\ldots,i_1$. 

\begin{lemma}[{\cite[Lemma 4.8]{kyncl_monotone}}]\label{lem:charoriginal}
Let $D$ be a drawing of $K_n$, and let $(i_n,\ldots,i_1)$ be a sequence of the vertices of $K_n$. Then: 
\begin{center}
$(i_n,\ldots,i_1)$ is an $x$-monotone sequence of $D$
\end{center}
\[
\iff
\]
\begin{center}
$(i_n,\ldots,i_1)$ is a shelling of $D$, $i_1$ and $i_n$ are incident with the unbounded cell of $D$, \\ and the path $i_n,\ldots,i_1$ does not cross itself in $D$.
\end{center}
\end{lemma}

We must point out that Lemma 4.8 in~\cite{kyncl_monotone} does not explicitly include the condition that $i_n$ and $i_1$ are incident with the outer face of $D$, but this is instead implicit in their formulation since their definition of a shelling sequence is slightly stronger than the original definition (namely Definition~\ref{def:shellable}), as it includes in addition that the first and last vertices of the sequence are incident with the unbounded face of the drawing.

Now the condition in Lemma~\ref{lem:charoriginal} that $(i_n,\ldots,i_1)$ is an $x$-monotone sequence of $D$ by definition means that (i) $i_1$ and $i_n$ are incident with the unbounded cell of $D$ and (ii) $D$ is weakly isomorphic to an $x$-monotone drawing with $x$-order $i_n,\ldots,i_1$. Thus the condition ``$i_1$ and $i_n$ are incident with the unbounded cell of $D$'' actually appears on both sides of Lemma~\ref{lem:charoriginal}, and so the lemma may be equivalently paraphrased as follows.

\vglue 0.35 cm
\noindent{\bf Lemma~\ref{lem:charoriginal} }(Equivalent formulation). {\em Let $D$ be a drawing of $K_n$, and let $(i_n,\ldots,i_1)$ be a sequence of the vertices of $K_n$. Then: 
\begin{center}
$D$ is weakly isomorphic to an $x$-monotone drawing with $x$-order $i_n,\ldots,i_1$
\end{center}
\[
\iff
\]
\begin{center}
$(i_n,\ldots,i_1)$ is a shelling of $D$ and the path $i_n,\ldots,i_1$ does not cross itself in $D$.
\end{center}
}
\vglue 0.35 cm

If we are only interested in knowing whether $D$ is monotone (that is, we recall, weakly isomorphic to an $x$-monotone drawing) and are not explicitly interested in the $x$-order of an $x$-monotone drawing weakly isomorphic to $D$ (if it exists), we note that this formulation of Lemma~\ref{lem:charoriginal} implies the following.

\begin{lemma}[Characterization of monotone drawings]\label{lem:cha}
A drawing $D$ of $K_n$ is monotone if and only if there is a shelling $(i_n,\ldots,i_1)$ of $D$ such that the path $i_n,\ldots,i_1$ does not cross itself in $D$.
\end{lemma}

\subsection{A Characterization of Monotonicity Using Partial Shelling Sequences}\label{sub:partial}

Our aim in this section is to give a version of Lemma~\ref{lem:cha} in terms of what we call {\em partial} shelling sequences, a weaker notion than shelling sequences. To motivate this concept we note that in order to verify whether a given sequence $(i_n,\ldots,i_1)$ is a shelling (so we can use Lemma~\ref{lem:cha}) we need to verify the shelling property for {\em all pairs $s,r$} of integers with $n \ge s > r \ge 1$. The notion of a partial shelling sequence considerably weakens this requirement, as it only involves verifying the shelling property for those pairs that include either $n$ or $1$.

\begin{definition}\label{def:partial}
{\sl Let $D$ be a drawing of $K_n$. We say that a sequence of vertices $S=(i_n , \ldots , i_1)$ is a {\em partial shelling sequence} (or simply a {\em partial shelling}) of $D$ if there is a cell $C$ of $D$ such that for all $n > r > 1$, vertex $i_r$ is on the boundary of the cell of $D(S,r,1)$ that contains $C$ and it is also on the boundary of the cell of $D(S,n,r)$ that contains $C$. We say that $C$ {\em witnesses} that $(i_n,\ldots,i_1)$ is a partial shelling sequence of $D$.}
\end{definition}

We note that if $D$ is a drawing of $K_n$ and $S$ is a shelling of $D$ witnessed by $C$, then clearly $S$ is also a partial shelling of $D$ witnessed by $C$.

The converse statement is not necessarily true, but it does hold if we add the condition that $i_n$ and $i_1$ are both on the boundary of $C$. Note that this property is implicitly satisfied by definition in shelling sequences, but not in partial shelling sequences.

This is the content of our next statement, which is closely related to~\cite[Observation 4.5]{kyncl_monotone}. This lemma will allow us to give a version of Lemma~\ref{lem:cha} (namely Lemma~\ref{lem:cha2} below) in terms of partial shelling sequences instead of shelling sequences.

\begin{lemma}\label{lem:shell}
Let $D$ be a drawing of $K_n$, let $C$ be a cell of $D$, and let $S = (i_n , \ldots , i_1)$ be a sequence of the vertices of $K_n$. Then $S$ is a shelling of $D$ witnessed by $C$ if and only if $S$ is a partial shelling of $D$ witnessed by $C$ and $i_1$ and $i_n$ are in the boundary of $C$.
\end{lemma}

\begin{proof}
The ``only if'' part follows trivially from the definitions of shellings and partial shellings.

For the ``if'' part we suppose that $S$ is a partial shelling of $D$ witnessed by a cell $C$ of $D$, and that $i_1$ and $i_n$ are in the boundary of $C$. In order to show that $S$ is a shelling witnessed by $C$ we let $s,r$ be integers such that $n \ge s > r \ge 1$, and prove that both $i_s$ and $i_r$ are on the boundary $\partial(C_{s,r})$ of the cell $C_{s,r}$ of $D(S,s,r)$ that contains $C$.

To prove this we first note that since $S$ is a partial shelling witnessed by $C$ it follows that $i_r$ is on the boundary $\partial(C_{n,r})$ of the cell $C_{n,r}$ of $D(S,n,r)$ that contains $C$. Similarly, $i_s$ is on the boundary $\partial(C_{s,1})$ of the cell $C_{s,1}$ of $D(S,s,1)$ that contains $C$. 

Since both $C_{n,r}$ and $C_{s,r}$ contain $C$ and $D(S,s,r) \subseteq D(S,n,r)$ it follows that $C_{n,r}\subseteq C_{s,r}$. Since $i_r$ is in $\partial(C_{n,r})$ it then follows that $i_r$ is in $\partial{(C_{s,r})}$, as claimed. A totally analogous argument shows that also $i_s$ is on $\partial(C_{s,r})$.
\end{proof}

\begin{lemma}[Characterizing monotonicity using partial shellings]\label{lem:cha2}
A drawing $D$ of $K_n$ is monotone if and only if there exist a sequence $(i_n,\ldots,i_1)$ of the vertices and a cell $C$ of $D$ such that:

\begin{enumerate}

\item[(m1)] $i_1$ and $i_n$ are in the boundary of $C$;

\item[(m2)] $(i_n,\ldots,i_1)$ is a partial shelling witnessed by $C$; and

\item[(m3)] the path $i_n,\ldots,i_1$ does not cross itself in $D$.

\end{enumerate}

\end{lemma}

\begin{proof}
The statement follows immediately by combining Lemmas~\ref{lem:cha} and~\ref{lem:shell}.
\end{proof}




\section{Proof of Lemma~\ref{lem:charac}}\label{sec:proofcharac}

\begin{proof}[Proof of Lemma~\ref{lem:charac}]
For the ``only if'' part suppose that $D$ is a monotone drawing of $K_n$. That is, $D$ is weakly isomorphic to an $x$-monotone drawing $D'$ of $K_n$. If $i_n,\ldots,i_1$ is the $x$-order of $D'$, then the $x$-monotonicity of $D'$ clearly implies that (M1)--(M3) hold in $D'$. Now using Gioan's theorem and that two drawings are weakly isomorphic if and only if their rotation systems are equivalent (see the discussion just before Corollary~\ref{cor:mon2}) it is not difficult to see that since (M1)--(M3) hold in $D'$ then they also hold in $D$. 

To prove the ``if'' part we suppose that $D$ is a drawing of $K_n$ such that there is a sequence $(i_n,i_{n-1},\ldots,i_1)$ of its vertices that satisfies (M1)--(M3), and we need to show that then $D$ is monotone. For convenience let us recall these properties:

\begin{enumerate}

\item[(M1)] there is a cell $C$ of $D$ that has both $i_1$ and $i_n$ in its boundary;

\item[(M2)] for each $r=n,n-1,\ldots,2$, there is an $i_r$-wedge whose vertex set is $\{i_{r-1},\ldots,i_1\}$; and 

\item[(M3)] for each $r=n,n-1,\ldots,3$, the edge $i_r i_{r-1}$ does not cross any edge with both endvertices in $\{i_{r-2},\ldots,i_1\}$. 
\end{enumerate}

In view of Lemma~\ref{lem:cha2}, in order to prove that $D$ is monotone it suffices to show that the sequence $(i_n,\ldots,i_1)$ and the cell $C$ satisfy (m1)--(m3) of Lemma~\ref{lem:cha2}. Now (m1) is exactly the same as (M1), and so the assumption that (M1) holds implies that (m1) holds. We also note that (M3) implies that the path $i_n,\ldots,i_1$ does not cross itself in $D$, and so (m3) also holds. 

Thus it only remains to show that (m2) holds, that is, that $(i_n,\ldots,i_1)$ is a partial shelling witnessed by $C$. Thus we need to show that if $r$ is an integer such that $n > r > 1$ then (a) vertex $i_r$ is on the boundary of the cell of $D(S,r,1)$ that contains $C$; and (b) $i_r$ is on the boundary of the cell of $D(S,n,r)$ that contains $C$.

We start with (a). For each $n > r > 1$ we let $C_r$ denote the cell of $D(S,r,1)$ that contains $C$. Thus our goal is to show that $i_r$ is on the boundary of $C_r$. We start with the case $r=n-1$. Since $i_n$ is in the boundary of $C$ and $C_{n-1}$ contains $C$, it follows that $C_{n-1}$ contains the union of the cells that have $i_n$ in their boundary. Since by (M3) the edge $i_n i_{n-1}$ is not crossed by any edge in $D(S,n-1,1)$ then the whole edge $i_n i_{n-1}$ must be contained in $C_{n-1}$, and from this it follows that $i_{n-1}$ must be in the boundary of $C_{n-1}$.

An inductive application of this argument, using for $r=n-2,\ldots,2$ that by (M3) the edge $i_{r+1} i_{r}$ is not crossed in $D(S,r,1)$, yields that $i_r$ is in the boundary of $C_r$ for every such $r$. Thus (a) follows.

Finally, to prove (b) we need to show that for each $r=n-1,\ldots,2$, vertex $i_r$ is incident with the cell $B_r$ of $D(S,n,r)$ that contains $C$. We note that the proof of (a) only relied in Property (M3), namely that for each $r=n-1,\ldots,2$, the edge $i_{r+1} i_{r}$ does not cross any ``forward'' edge (an edge $i_t i_s$ where $t$ and $s$ are both smaller than $r$). It is easy to see that if the analogous property in the opposite (``backward'') direction holds, then totally analogous arguments imply (b). Therefore in order to prove (b) it suffices to show that for each $r=n-1,\ldots,2$, the edge $i_r i_{r-1}$ does not cross any edge $i_t i_s$ where $t$ and $s$ are both {\em larger} than $r$. As we shall see, this is the only place in the proof in which we make use of (M2).

Seeking a contradiction, suppose that for some $r=n-1,\ldots,2$ the edge $i_r i_{r-1}$ crosses an edge $i_t i_s$ with $t > r$ and $s > r$. We choose the labels $s$ and $t$ so that $t > s$. Note that $t$ cannot be $s+1$, as this would mean that $i_ti_s=i_{s+1} i_{s}$ crosses $i_r i_{r-1}$: since $r$ and $r-1$ are  smaller than $s$, this would contradict (M3) for the edge $i_{s+1} i_{s}$. Thus $n \ge t > s+1 > s > r > r-1$. 

To derive the contradiction we start by noting that up to equivalence there is only one drawing of the $K_4$ induced by $\{i_t,i_s,i_r,i_{r-1}\}$ where $i_t i_s$ crosses $i_r i_{r-1}$, namely the one in Figure~\ref{fig:445}(a) (perhaps with the labels of $i_r$ and $i_{r-1}$ interchanged, but this is irrelevant for the upcoming arguments). Thus without loss of generality we may assume that the restriction of $D$ to this $K_4$ is exactly as in that figure.

\begin{figure}[ht!]
\def\ta#1{{\Scale[1.6]{#1}}}
\def\ta#1{{\Scale[5.5]{#1}}}
\def\Somea{{\Scale[5.5]{\text{\rm (a)}}}}
\def\Someb{{\Scale[5.5]{\text{\rm (b)}}}}
\def\Somec{{\Scale[5.5]{\text{\rm (c)}}}}
\centering
\scalebox{0.15}{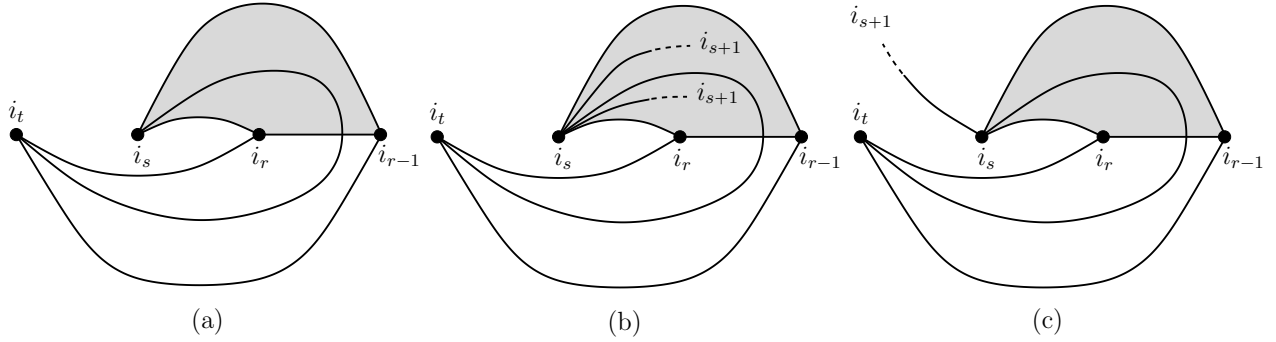}
\caption{Illustration of the conclusion of the proof of Lemma~\ref{lem:charac}.}
\label{fig:445}
\end{figure}

We claim that the edge $i_{s+1} i_s$ must leave $i_s$ in one of the two gray regions in Figure~\ref{fig:445}(b), as hinted in the illustration. This will provide the required contradiction, as this implies that (some edge in) the path $i_t,i_{t-1},\ldots,i_{s+1},i_s$ crosses the edge $i_r i_{r-1}$, contradicting (M3).

To prove that the edge $i_{s+1} i_s$ must leave $i_s$ in one of the gray regions, by way of contradiction suppose that this is not the case. Then $i_{s+1} i_s$ must leave $i_s$ as shown in part (c) of that figure, and so the rotation at vertex $i_s$ contains the vertices $i_{r-1},i_{t}, i_{r},i_{s+1}$ in this cyclic order. Since $t$ and $s+1$ are both larger than $s$, and $r$ and $r-1$ are both smaller than $s$, it follows that no $i_s$-wedge has $\{i_{s-1},\ldots,i_1\}$ as its vertex set, contradicting (M2).
\end{proof}

\section{Proof of Lemma~\ref{lem:n5}}\label{sec:algo}

The algorithm claimed in Lemma~\ref{lem:n5} relies crucially on three data structures that allow us to perform certain queries in constant time. The existence of these data structures is stated in the next three lemmas, whose proofs are deferred to Section~\ref{sec:proofsds1ds2ds3}.

We encourage the reader to skip these lemmas for the time being and come back to them until they are invoked, at the very end of the proof of Lemma~\ref{lem:n5}.

\begin{lemma}[First data structure]\label{lem:ds1}
Given a drawing $D$ of $K_n$, we can construct in $O(n^2)$ time a data structure to answer the following queries in constant time.
\begin{enumerate}

\item[(Q1)] Given a vertex $k$ and a wedge $\ww$, is $k$ in $\ww$?

\item[(Q2)] Given two edges of $K_n$, do these edges cross each other in $D$?

\end{enumerate}
\end{lemma}

\begin{lemma}[Second data structure]\label{lem:ds2}
Suppose that the data structure from Lemma~\ref{lem:ds1} has been constructed. Given a drawing $D$ of $K_n$, we can construct in $O(n^5)$ time a data structure to answer the following query in constant time. 

Let $\ww$ be a wedge, and let $k$ be a vertex in $\ww$. Is there a wedge $\zz$ with hub $k$ such that $\set{\ww}=\set{\zz}\cup\{k\}$?

In addition, the data structure provides $\zz$, if it exists. (We note that if $\zz$ exists, then it is unique.)
\end{lemma}

\begin{lemma}[Third data structure]\label{lem:ds3}
Suppose that the data structure from Lemma~\ref{lem:ds1} has been constructed. Given a drawing $D$ of $K_n$, we can construct in $O(n^4)$ time a data structure to answer the following query in constant time.

Let $\ww$ be a wedge, let $i$ be its hub vertex, and let $k$ be a vertex in $\ww$. Is the edge $i\, k$ crossed by some edge that has both endvertices in $\ww\setminus\{k\}$?
\end{lemma}

\subsection{Proof of Lemma~\ref{lem:n5}}\label{sub:proofn5}

We recall that Lemma~\ref{lem:n5} claims the existence of an algorithm that, given a drawing $D$ of $K_n$, verifies in $O(n^5)$ time whether there is a sequence $(i_n,\ldots,i_1)$ of the vertices that satisfies the following properties:

\begin{enumerate}

\item[(M1)] there is a cell $C$ of $D$ that has both $i_1$ and $i_n$ in its boundary;

\item[(M2)] for each $r=n,n-1,\ldots,2$, there is an $i_r$-wedge whose vertex set is $\{i_{r-1},\ldots,i_1\}$; and

\item[(M3)] for each $r=n,n-1,\ldots,3$, the edge $i_r i_{r-1}$ does not cross any edge with both endvertices in $\{i_{r-2},\ldots,i_1\}$.

\end{enumerate}


To deal with (M2) and (M3) we extend these two properties to sequences with fewer than $n$ vertices, as follows.

\begin{definition}[Good sequences of vertices]\label{def:goodseq}
{\sl Let $D$ be a drawing of $K_n$, and let $t\in\{n,n-1,\ldots,2\}$. We say that a sequence $(i_t,i_{t-1},\ldots,i_1)$ of vertices is {\em good in} $D$ if the following hold:}

\begin{enumerate}

\item[(G1)] {\sl for each $r=t,t-1,\ldots,2$, there is an $i_r$-wedge whose vertex set is $\{i_{r-1},\ldots,i_1\}$; and}

\item[(G2)] {\sl the path $i_t,\ldots,i_1$ does not cross itself in $D$.}

\end{enumerate}

\end{definition}

Note that for the case $t=n$ Properties (G1) and (G2) become exactly (M2) and (M3). 

\begin{proof}[Proof of Lemma~\ref{lem:n5}]
Let $D$ be a drawing of $K_n$. We start by constructing, for later use in this proof, the three data structures from Lemmas~\ref{lem:ds1}--~\ref{lem:ds3}. According to these lemmas the three data structures can be constructed in $O(n^5)$ time.

Having these data structures at our disposition, our goal is to show the existence of an algorithm that finds out in $O(n^5)$ time whether there is a sequence $(i_n,\ldots,i_1)$ that satisfies (M1)--(M3). 

At a high level our strategy consists of stating a sequence of claims (Claims A--E). We argue that Claim E $\implies$ Claim D $\implies$ Claim C $\implies$ Claim B $\implies$ Claim A $\implies$ Lemma~\ref{lem:n5}, and finish the proof by proving Claim E.

For the rest of the proof the first vertex in a sequence of vertices is the {\em head} of the sequence, and the last one is its {\em tail}.

\vglue 0.3 cm
\noindent{\bf Claim A. }{(Implies Lemma~\ref{lem:n5}). }{\sl There is an $O(n^5)$ time algorithm that finds all pairs $i,j$ of vertices such that there is a good sequence of length $n$ with head $i$ and tail $j$.}
\vglue 0.3 cm

To see that Claim A indeed implies the lemma, we run in $O(n^5)$ time the algorithm given by Claim A and find all pairs $i,j$ such that there is a good sequence of length $n$ with head $i$ and tail $j$. As we noted immediately after Definition~\ref{def:goodseq}, these are precisely the sequences $(i=i_n,i_{n-1},\ldots,i_2,i_1=j)$ that satisfy (M2) and (M3). Thus in order to verify whether there is a sequence that satisfies (M1)--(M3) it suffices to check, for each such pair $i,j$, whether there is a cell that has both $i$ and $j$ in its boundary. Regardless of the specific format in which we are given $D$ (its cell structure or its extended rotation system) it is easy to see that this can be tested for each fixed pair in $O(n^3)$ time. Since there are $O(n^2)$ pairs to test, it follows that this last step can also be performed in $O(n^5)$ time. 

Thus Claim A implies Lemma~\ref{lem:n5}. We now note that since there are $n$ vertices $j$ that can be the tail vertex of a good sequence of length $n$, in order to prove Claim A it suffices to show the following.   

\vglue 0.3 cm
\noindent{\bf Claim B. }{(Implies Claim A). }{\sl Let $j$ be a fixed vertex. There is an $O(n^4)$ time algorithm that finds all vertices $i$ such that there is a good sequence of length $n$ with head $i$ and tail $j$.}
\vglue 0.3 cm

We now note that the following statement is strictly stronger than Claim B:

\vglue 0.3 cm
\noindent{\bf Claim C. }{(Implies Claim B). }{\sl Let $j$ be a fixed vertex. There is an $O(n^4)$ time algorithm that finds for {\rm every} $t=2,\ldots,n$ all vertices $i$ such that there is a good sequence of length $t$ with head $i$ and tail $j$.}
\vglue 0.3 cm

In the rest of the proof we also attribute the goodness property (which applies to sequences) to wedges by saying that a wedge $\ww$ with hub $i$ is {\em good\,} if there is a good sequence $(i=i_t,i_{t-1},\ldots,i_2,i_1)$ (that is, with head vertex $i$) where $\set{\ww}= \{i_{t-1},\ldots,i_1\}$. We say that the wedge $\ww$ is {\em good with tail $i_1$}.

We note that a good wedge can be of any size from $1$ to $n-1$. Indeed, a wedge with hub $i$ that only contains one vertex $j$ is always good with tail $j$, as the sequence $(i,j)$ is clearly good.

Since wedges are linear permutations it seems worth emphasizing that saying that a wedge $\ww$ is good with tail $j$ does {\em not} mean that $j$ is necessarily the last vertex of the linear permutation $\ww$. It must also be noted that a given wedge may be good with more than one tail vertex, as there might be more than one good sequence that witnesses the goodness of the wedge.

The definition of a good wedge implies that for $t=2,\ldots,n$ there is a good sequence of length $t$ with head $i$ and tail $j$ if and only if there is a wedge of size $t-1$ with hub $i$ that is good with tail~$j$. From this it follows that the next statement implies Claim C:

\vglue 0.3 cm
\noindent{\bf Claim D. }{(Implies Claim C). }{\sl Let $j$ be a fixed vertex. There is an $O(n^4)$ time algorithm that finds for {\rm every} $s=1,\ldots,n-1$ all the wedges of size $s$ that are good with tail $j$.}
\vglue 0.3 cm

We now state the following:

\vglue 0.3 cm
\noindent{\bf Claim E. }{(Implies Claim D). }{\sl Let $j$ be a fixed vertex. Suppose that for some $s\in\{1,\ldots,n-2\}$ we have found all the wedges of size $s$ that are good with tail $j$. Then there is an $O(n)$ time algorithm that verifies whether a given wedge of size $s+1$ is good with tail $j$.}
\vglue 0.3 cm

To see that Claim E indeed implies Claim D we start by noting that it is trivial to find all the wedges of size $1$ that are good with tail $j$: these are simply all the wedges of size $1$ whose only vertex is $j$.

Suppose now that Claim E holds. Note that for each $s=1,\ldots,n-2$ there are $O(n^2)$ wedges of size $s+1$. Thus Claim E implies that if for some $s\in\{1,\ldots,n-2\}$ all the wedges that are good with tail $j$ have been found then all the wedges of size $s+1$ that are good with tail $j$ can be found in $O(n^3)$ time (each of the $O(n^2)$ wedges gets tested in $O(n)$ time). Therefore for each fixed $s\in\{1,\ldots,n-1\}$ all the wedges of size $s$ that are good with tail $j$ can be found in $O(n^3)$ time, and from this Claim D immediately follows.

We have thus shown that Claim E $\implies$ Claim D $\implies$ Claim C $\implies$ Claim B $\implies$ Claim A $\implies$ Lemma~\ref{lem:n5}. We finish the proof by proving Claim E.

Let $j$ be a fixed vertex, and let $s\in\{1,\ldots,n-2\}$. Suppose that we have found all the wedges of size $s$ that are good with tail $j$, and let $\ww$ be a wedge of size $s+1$. Our goal is to show that we can test in $O(n)$ time whether $\ww$ is good with tail $j$.

Using the definition of a good wedge it is not difficult to verify that $\ww$ is good with tail $j$ if and only if there is a wedge $\zz$ of size $s$ such that:

\begin{enumerate}

\item[(i)] the hub $k$ of $\zz$ is in $\ww$; 

\item[(ii)] $\set{\ww}=\set{\zz}\cup \{k\}$;

\item[(iii)] the wedge $\zz$ is good with tail $j$; and 

\item[(iv)] the edge that joins the hub $i$ of $\ww$ to $k$ does not cross any edge with both endvertices in $\zz$. 

\end{enumerate}

In view of this, in order to prove Claim E it suffices to show that the existence of a wedge $\zz$ that satisfies (i)--(iv) can be decided in $O(n)$ time.

In order to achieve this we consider each vertex $k$ different from $i$ and test whether ($\dag$) {\em there is a wedge $\zz$ with hub $k$ that satisfies (i)--(iv)}. We claim that for each fixed vertex $k$ we can decide ($\dag$) in constant time. Since (obviously) every wedge has a hub, from this it will follow that one can check in $O(n)$ time (as in the worst case scenario there are $n-1$ vertices $k$ to consider) whether there is a wedge $\zz$ that satisfies (i)--(iv), thus finishing the proof of Claim E.

Thus we let $k$ be any vertex distinct from $i$. We recall that we started the proof by constructing the data structures from Lemmas~\ref{lem:ds1}--~\ref{lem:ds3}. Thus in particular it follows from Lemma~\ref{lem:ds1}(Q1) that we can test in constant time whether $k$ is in $\ww$. If not then clearly $k$ cannot be the hub of a wedge $\zz$ that satisfies (i), and so the answer to ($\dag$) is no. If yes, then we use Lemma~\ref{lem:ds2} to find out whether there is a wedge $\zz$ with hub $k$ that satisfies (ii). Again, if the answer is no then the answer to ($\dag$) is no. If the answer is yes, we note that the query from Lemma~\ref{lem:ds2} also returns the (unique) wedge $\zz$ that satisfies (ii). Thus we move on to test (iii) and (iv) with this wedge $\zz$.

By assumption all wedges of size $s$ that are good with tail $j$ have been determined, and so it can be verified in constant time whether (iii) holds for this wedge $\zz$. Again, if the answer is no then the answer to ($\dag$) is no. If the answer is yes it only remains to verify whether (iv) holds, and in view of Lemma~\ref{lem:ds3} this can be done in constant time. If the answer is no then the answer to ($\dag$) is no, and if the answer is yes then the answer to ($\dag$) is yes.
\end{proof}

\section{Proof of Lemmas~\ref{lem:ds1}--\ref{lem:ds3}}\label{sec:proofsds1ds2ds3}

We recall (see the remark immediately after Theorem~\ref{thm:mon}) that we may assume that we are given as input either the extended rotation system of $D$ or its cell structure, and from this we can easily obtain the rotation system of $D$. So we may as well assume that we are given from the onset the rotation system of $D$.

\subsection{Proof of Lemma~\ref{lem:ds1}}\label{sec:ds1}

We note that even though for each vertex $i\in[n]$ its rotation $\cca{i}$ in $D$ is a cyclic permutation, the natural data type to encode it for computational purposes is as a (linear) array $\ora{i}$. Thus in the next proof we assume that this is the data type in which we are given the rotation at each vertex.

\begin{proof}[Proof of Lemma~\ref{lem:ds1}]
To obtain the data structure we construct for each $i\in[n]$ an array $\Ind(i)$ that stores the index in $\ora{i}$ of each vertex distinct from $i$. Clearly for each fixed $i$ the array $\Ind(i)$ can be constructed in $O(n)$ time, and so the entire set of arrays $\Ind(1),\ldots,\Ind(n)$ is constructed in $O(n^2)$ time. We claim that the data structure that consists of this set of arrays allows us to answer queries (Q1) and (Q2) in Lemma~\ref{lem:ds1} in constant time.

Regarding (Q1) we first note that once we have $\Ind(i)$ for some $i=1,\ldots,n$ then every $i$-wedge can be stored using only two integers, namely the indices of the initial vertex and of the final vertex of the wedge. Suppose that we are given a vertex $k$ and an $i$-wedge $\ww$, and want to know whether $k$ is in $\ww$. Since we know the indices in $\Ind(i)$ of $k$ and of the initial and final vertices of $\ww$ it follows that we can answer in constant time whether $k$ is in $\ww$.

Finally, to handle (Q2) suppose that we are given two edges $i\, j$ and $k\,\ell$ and want to know whether they cross each other in $D$. It is not difficult to see that using $\Ind(i), \Ind(j), \Ind(k)$, and $\Ind(\ell)$ we can obtain in constant time the rotation system of the $K_4$ induced by the vertices $i,j,k$ and $\ell$. Since from a rotation system of $K_4$ a quick inspection reveals whether or not a given pair of its edges cross each other (see~\cite{enukyncl,howmany}), it follows that indeed finding out whether these edges cross each other can be determined in constant time, as required.
\end{proof}

\subsection{Proof of Lemma~\ref{lem:ds2}}\label{sec:ds2}

We recall that in the context of Lemma~\ref{lem:ds2} we have a vertex $k$ in a wedge $\ww$, and one is interested in knowing whether there is a $k$-wedge $\zz$ such that $\set{\ww}=\set{\zz}\cup \{k\}$. Equivalently, we want to know whether there is a $k$-wedge $\zz$ such that $\set{\zz}=\set{\ww}\setminus \{k\}$.

\subsubsection{Shortest Wedges that Contain a Given Vertex Set}

The first step towards the construction of a data structure that answers this query in constant time is an observation on the shortest wedges that contain a given vertex set. In order to discuss this properly let $k$ be a vertex, and let $U$ be a set of vertices that does not contain $k$. Needless to say, regardless of $k$ and $U$ there always exist $k$-wedges that contain all the vertices in $U$ (for instance, any $k$-wedge of size $n-1$ certainly contains $U$). We shall use $\kint{k}{U}$ to denote the collection of all the shortest $k$-wedges that contain $U$.

The number of $k$-wedges in $\kint{k}{U}$ greatly depends on $k$ and $U$. For instance, if $|U|\le n-2$ and $U$ happens to be the set of vertices of a $k$-wedge then clearly this $k$-wedge is the only element of $\kint{k}{U}$. In the opposite end if $|U|=n-1$ (that is, $U$ contains all the vertices except for $k$) then $\kint{k}{U}$ consists of all the $k$-wedges of size $n-1$ (and there are $n-1$ such $k$-wedges). In any case, regardless of $k$ and $U$ necessarily $\kint{k}{U}$ contains $O(n)$ wedges, as there are $n-1$ $k$-wedges of each fixed size.

Back to the main discussion, we have a vertex $k$ in a wedge $\ww$. If we let $U:=\set{\ww}\setminus\{k\}$, we are interested in knowing whether there is a $k$-wedge $\zz$ such that $\set{\zz}=U$. As it happens, this problem is closely related to $\kint{k}{U}$ via the following remark, which is not difficult to verify.

\begin{observation}\label{obs:kint}
Let $k$ be a vertex, and let $U$ be a set of at most $n-2$ vertices that does not contain $k$. Let $\zz$ be a $k$-wedge. Then $\set{\zz}=U$ if and only if $\kint{k}{U}=\{\zz\}$ and $|\set{\zz}|=|U|$.
\end{observation}

A key result behind the proof of Lemma~\ref{lem:ds2} is that in general we do not need to calculate $\kint{k}{U}$ from scratch for a given set $U$. Instead, we can use the next statement to obtain $\kint{k}{U}$ in a recursive way.

\begin{claim}\label{cla:aux1}
Let $k$ and $\ell$ be distinct vertices, and let $U$ be a set of vertices that contains neither $k$ nor $\ell$. If $\kint{k}{U}$ is known, then $\kint{k}{U\cup\{\ell\}}$ can be obtained in $O(n)$ time.
\end{claim}

\begin{proof}
We start by noting that if $|U\cup\{\ell\}|=n-1$ then, as we noted above, $\kint{k}{U\cup\{\ell\}}$ consists of the $n-1$ $k$-wedges of size $n-1$, and so there is nothing else to be done for this case. Thus we may assume that $|U\cup\{\ell\}| \le n-2$.

Clearly, every $k$-wedge in $\kint{k}{U\cup\{\ell\}}$ must contain as a subwedge some $k$-wedge in $\kint{k}{U}$. Thus it suffices to consider one by one each $k$-wedge $\vv$ in $\kint{k}{U}$, and ($*$) {\sl find the shortest $k$-wedges that contain $\set{\vv}\cup\{\ell\}$}. Indeed, $\kint{k}{U\cup\{\ell\}}$ will simply consist of the shortest $k$-wedges we encountered as we ran this process for all $\vv\in\kint{k}{U}$. Since $\kint{k}{U}$ consists of $O(n)$ $k$-wedges, in order to prove the claim it suffices to show that if $\vv$ and $\ell$ are given, then ($*$) can be performed in constant time. In particular, as we are about to see, running ($*$) when $\vv$ and $\ell$ are given yields that there are at most two shortest $k$-wedges that contain $\set{\vv}\cup\{\ell\}$.

Suppose then that $\vv$ and $\ell$ are given. We note that if $\ell$ is already in $\vv$ (a test that gets done in constant time, in view of Lemma~\ref{lem:ds1}(Q1)) then $\vv$ is the only shortest $k$-wedge that contains $\set{\vv}$ and $\ell$, and so we are done. 

Suppose finally that vertex $\ell$ is not in $\vv$. Let $i$ (respectively, $j$) be the first (respectively, last) vertex in $\vv$. Let $\vv_{i,\ell}$ be the $k$-wedge with initial vertex $i$ and final vertex $\ell$, and let $\vv_{\ell,j}$ be the $k$-wedge with initial vertex $\ell$ and final vertex $j$. These two $k$-wedges clearly contain $\set{\vv}$ and $\ell$. If one of them is shorter than the other then $\kint{k}{\set{\vv}\cup\{\ell\}}$ consists solely of this $k$-wedge, and if they have the same size then $\kint{k}{\set{\vv}\cup\{\ell\}}$ consists of these two $k$-wedges. We finally note that $\vv_{i,\ell}$ and $\vv_{\ell,j}$ get determined in constant time and their sizes are also calculated in constant time, and so $\kint{k}{\set{\vv}\cup\{\ell\}}$ gets determined in constant time.
\end{proof}

\subsubsection{Proof of Lemma~\ref{lem:ds2}}


In the proof of the lemma we use the following terminology. If $\ww$ is a wedge and $k$ is a vertex in $\ww$ then we say that $(\ww,k)$ is a {\em valid pair}. The number of vertices in $\ww$ is the {\em rank} of the pair $(\ww,k)$.

\begin{proof}[Proof of Lemma~\ref{lem:ds2}]

We claim that in order to prove Lemma~\ref{lem:ds2} it suffices to show that ($*$) {\sl in $O(n^5)$ time we can construct a data structure that stores $\kint{k}{\set{\ww}\setminus\{k\}}$ for all valid pairs $(\ww,k)$, and together with $\kint{k}{\set{\ww}\setminus\{k\}}$ also stores the size of each $k$-wedge in this collection} (since they are the shortest $k$-wedges that contain $\ww\setminus\{k\}$, they all have the same size).

To see that having this data structure proves Lemma~\ref{lem:ds2}, suppose that it has been built, let $\ww$ be a wedge, and let $k$ be a vertex in $\ww$. In the context of Lemma~\ref{lem:ds2} we are asked the following question: is there a $k$-wedge $\zz$ such that $\set{\ww}=\set{\zz}\cup \{k\}$ (equivalently, $\set{\zz}=\set{\ww}\setminus \{k\}$)? Letting $U:=\set{\ww}\setminus\{k\}$, this is equivalent to asking: is there a $k$-wedge $\zz$ such that $\set{\zz}=U$?

In view of Observation~\ref{obs:kint}, this is equivalent to asking: is there a $k$-wedge $\zz$ such that $\kint{k}{U}=\{\zz\}$ and $|\set{\zz}|=|U|$? Retrieving from the constructed data structure (in constant time) $\kint{k}{\set{\ww}\setminus\{k\}}$ and the size of each $k$-wedge in this collection answers this question: indeed, it suffices to check whether $\kint{k}{\set{\ww}\setminus\{k\}}$ consists of a single $k$-wedge and, if so, whether the size of this $k$-wedge is the same as the size of $U$. Thus the query can be answered in constant time, as claimed. Moreover, if the answer is yes then we have obtained $\zz$ as well, as it is also claimed in the lemma. Thus in order to prove the lemma it suffices to show ($*$). 

To prove ($*$) we start by noting that if $(\ww,k)$ is a valid pair of rank $2$ (that is, $\ww$ has size $2$), then $\kint{k}{\ww\setminus\{k\}}$ consists of a single $k$-wedge which can be computed in constant time. Indeed, in this case $\set{\ww\setminus\{k\}}$ has only a vertex $\ell$,  and $\kint{k}{\{\ell\}}$ consists of the (unique) $k$-wedge of size $1$ whose only vertex is $\ell$. 

Now since there are $O(n^2)$ wedges of size $2$ and for each such wedge there are two possible ways to choose $k$, we conclude that $\kint{k}{\set{\ww}\setminus\{k\}}$ can be computed in $O(n^2)$ time for all valid pairs $(\ww,k)$ of rank $2$.

To construct the data structure claimed in ($*$) for valid pairs of rank $>2$ we proceed as follows. We show that ($\ddag$) {\em if the data structure has been built for all valid pairs of rank $r$ for some $2 \le r \le n-2$, then $\kint{k}{\ww\setminus\{k\}}$ can be found in $O(n)$ time for each valid pair $(\ww,k)$ of rank $r+1$}. Since there are $O(n^2)$ wedges of size $r+1$ then there are $O(n^3)$ valid pairs $(\ww,k)$ of rank $r+1$, and so from ($\ddag$) it will follow that $\kint{k}{\ww\setminus\{k\}}$ can be determined for {\em all} valid pairs $(\ww,k)$ of rank $r+1$ in $O(n^4)$ time. Since there are $O(n)$ possible values of $r$, this will show that the whole data structure can be constructed in $O(n^5)$ time, thus proving ($*$).

To prove ($\ddag$) we assume that the data structure has been built for all valid pairs of rank $r$ for some $2\le r \le n-2$, and let $(\ww,k)$ be a valid pair of rank $r+1$. Let $i$ be the hub of $\ww$ and let $j$ be its initial vertex, so that $\ww=\H{i}{j}{r+1}$. We let $j'$ be the vertex that succeeds $j$ in $\H{i}{j}{r+1}$, and let $j''$ be the last vertex of $\H{i}{j}{r+1}$. Thus the $i$-wedge $\H{i}{j}{r+1}$ is of the form $\perm{j,j',\ldots,j''}$. Now since vertex $k$ is in $\H{i}{j}{r+1}$ then it is in $\H{i}{j}{r}$ unless $k=j''$, and it is in $\H{i}{j'}{r}$ unless $k=j$. In particular $k$ is either in $\H{i}{j}{r}$ or in $\H{i}{j'}{r}$. We assume that $k$ is in $\H{i}{j}{r}$, as the latter possibility is handled in a totally analogous manner.

Recall that our goal is to obtain $\kint{k}{\ww\setminus\{k\}}$. Now $\ww\setminus\{k\}=\H{i}{j}{r+1}\setminus\{k\}$ is the union of $\H{i}{j}{r}\setminus\{k\}$ and $\{j''\}$. Since the data structure has been constructed for all valid pairs of rank $r$ then $\kint{k}{\H{i}{j}{r}\setminus\{k\}}$ is known, and so it follows from Claim~\ref{cla:aux1} that $\kint{k}{\H{i}{j}{r}\setminus\{k\}\cup\{j''\}}=\kint{k}{\H{i}{j}{r+1}\setminus\{k\}}=\kint{k}{\ww\setminus\{k\}}$ can be obtained in $O(n)$ time, as claimed in ($\ddag$).
\end{proof}

\subsection{Proof of Lemma~\ref{lem:ds3}}\label{sec:ds3}

Before we proceed to the proof of the lemma let us define and discuss a function that is used in our arguments, namely the function $\kl$ (for ``closest'').

This function is defined as follows. Let $(i,k,\ell)$ be a triple of distinct vertices. We say that a vertex $m$ is $(i,k,\ell)$-{\em active} if the edge $\ell\,m$ crosses the edge $i\,k$. If no $(i,k,\ell)$-active vertices exist, then we let $\kl(i,k,\ell)=0$. Otherwise, we traverse the cyclic permutation $\cca{i}$ in reverse order starting from $\ell$ until we find a vertex $m$ such that the edge $\ell\, m$ crosses $i\,k$. We then let $\kl(i,k,\ell)$ be the number of vertices we needed to explore, including $m$, to reach $m$. Thus in this case $\kl(i,k,\ell)$ can be any integer in $\{1,\ldots,n-2\}$ (it cannot be greater than $n-2$ because $\cca{i}$ has $n-1$ vertices, and so starting the traversal of $\cca{i}$ from vertex $\ell$ there are $n-2$ vertices to explore).

Consider for instance the triple $(3,1,6)$ in the drawing of $K_7$ in Figure~\ref{fig:150}. In order to find $\kl(3,1,6)$ we start by getting $\cca{3}$, which in this case is $\cycper{2,6,5,7,4,1}$. To obtain $\kl(3,1,6)$ we start traversing this cyclic permutation in reverse order starting from $6$. Starting from $6$ the first vertex to consider is $2$. Since the edge $6\, 2$ does not cross the edge $3\, 1$ we move on to consider vertex $1$. Since $6\, 1$ (obviously) does not cross $3\, 1$ either, we move on to consider vertex $4$. Now $6\, 4$ does not cross $3\, 1$ either, and so next we consider vertex $7$. Again, $6\, 7$ does not cross $3\, 1$, and so we move on to vertex $5$. The edge $5\, 1$ does cross $3\, 1$, and so we stop the process and count the number of vertices we considered until we reached vertex $5$. We considered $2,1,4,7,5$, that is, five vertices in total. Therefore in this example we have $\kl(3,1,6)=5$.

\begin{proof}
By assumption the data structure from Lemma~\ref{lem:ds1} has been constructed. Therefore for each fixed triple $(i,k,\ell)$ we can compute $\kl(i,k,\ell)$ by performing $O(n)$ queries (of the type (Q2) in Lemma~\ref{lem:ds1}), each of which is done in constant time. Thus for each fixed triple $(i,k,\ell)$ we can compute $\kl(i,k,\ell)$ in $O(n)$ time. We store this information in a data structure for later use, and move on with the rest of the proof.

As in the proof of Lemma~\ref{lem:ds2}, if $\ww$ is a wedge and $k$ is a vertex in $\ww$ then the pair $(\ww,k)$ is {\em valid}, and the size of $\ww$ is the {\em rank} of the valid pair. Let $(\ww,k)$ be a valid pair, and let $i$ be the hub of $\ww$. We say that $(\ww,k)$ is {\em good} if no edge with both endvertices in $\ww$ crosses the edge $i\, k$. Otherwise we say that $(\ww,k)$ is {\em bad}. 

Using this terminology, Lemma~\ref{lem:ds3} claims the existence of a data structure that contains for each valid pair $(\ww,k)$ the information of whether it is good or bad. 

We construct this data structure using dynamic programming. We first note that for valid pairs of rank $2$ there is nothing to do. Indeed, suppose that $(\ww,k)$ is a valid pair where $\ww$ has size $2$ and its hub vertex is $i$. Then there is only one vertex in $\ww$ other than $k$, and so obviously no edge with both endvertices in $\ww$ can cross the edge $i\,k$. 

To construct the data structure for all valid pairs of rank $>2$ we show that ($\ddag$) {\em if the data structure has been constructed for all valid pairs of rank $r$ for some $2 \le r \le n-2$, then for each valid pair of rank $r+1$ we can determine whether it is good or bad in constant time}. It is easy to verify that there are in total $O(n^4)$ valid pairs of rank greater than $2$ (that is, greater than $2$ and smaller than $n$), and so ($\ddag$) implies that the whole data structure can be constructed in time $O(n^4)$, as claimed.

To prove ($\ddag$) we assume that the data structure has been constructed for all valid pairs of rank $r$ for some $r$, $2 \le r \le n-2$, and let $(\ww,k)$ be a valid pair of rank $r+1$. Let $i$ be the hub vertex of $\ww$, and let $j$ be the first vertex of $\ww$, so that $\ww=\H{i}{j}{r+1}$. We let $j'$ be the vertex that succeeds $j$ in $\H{i}{j}{r+1}$, and let $j''$ be the last vertex of $\H{i}{j}{r+1}$. Thus $\H{i}{j}{r+1}$ is of the form $\perm{j,j',\ldots,j''}$. Since by assumption vertex $k$ is in $\H{i}{j}{r+1}$ then it is in $\H{i}{j}{r}$ unless $k=j''$, and it is in $\H{i}{j'}{r}$ unless $k=j$. In particular $k$ is either in $\H{i}{j}{r}$ or in $\H{i}{j'}{r}$. We assume that $k$ is in $\H{i}{j}{r}$, as the latter possibility is handled in a totally analogous manner.

In order to show ($\ddag$) we need to show that it is possible to decide in constant time whether there is an edge with both endvertices in $\H{i}{j}{r+1}$ that crosses the edge $i\,k$. Now there are two possibilities for an edge with both endvertices in $\H{i}{j}{r+1}$: either (i) it has both endvertices in $\H{i}{j}{r}$; or (ii) it has one endvertex in $\H{i}{j}{r}$ and its other endvertex is $j''$. 

Since $k$ is in $\H{i}{j}{r}$ and the data structure has been constructed for valid pairs of rank $r$ it follows that we can tell in constant time whether some edge that satisfies (i) crosses the edge $i\, k$. Regarding the edges that satisfy (ii) we look at $\kl(i,k,j'')$. If it is $0$ then no edge $j''$ as an endvertex crosses $i\,k$, and so we are done. Otherwise $\kl(i,k,j'')$ is some positive integer $t$. The definition of $\kl$ implies that an edge satisfying (ii) crosses $i\,k$ if and only if $t \le r$. Since this inequality can obviously be verified in constant time, we conclude that we can tell in constant time whether some edge with both endvertices in $\H{i}{j}{r+1}$ crosses the edge $i\, k$. That is, we can decide in constant time whether $(\ww,k)$ is good or bad, and so ($\ddag$) follows.
\end{proof}


\section*{Acknowledgments} 
We thank Bernardo M. \'Abrego, Silvia Fern\'andez-Merchant, and Pedro Ramos for valuable discussions. We thank C\'esar Hern\'andez-V\'elez for his valuable input on an earlier version of this paper. The first author was supported by the ESF EUROCORES programme EuroGIGA-ComPoSe, Austrian Science Fund (FWF): I 648-N18.. The second author was supported by the Austrian Science Fund (FWF): P23629-N18 ``Combinatorial Problems on Geometric Graphs''. The third author was supported by a Schr\"odinger fellowship of the Austrian Science Fund (FWF): J-3847-N35.

\bibliographystyle{amsplain}
\bibliography{bibliography}

\end{document}

%% file: 380b.pdf_t
\begin{picture}(0,0)%
\includegraphics{380b.pdf}%
\end{picture}%
\setlength{\unitlength}{4144sp}%
\begingroup\makeatletter\ifx\SetFigFont\undefined%
\gdef\SetFigFont#1#2#3#4#5{%
  \reset@font\fontsize{#1}{#2pt}%
  \fontfamily{#3}\fontseries{#4}\fontshape{#5}%
  \selectfont}%
\fi\endgroup%
\begin{picture}(11415,7442)(26581,-11252)
\put(26596,-7891){\makebox(0,0)[lb]{\smash{{\SetFigFont{20}{24.0}{\familydefault}{\mddefault}{\updefault}{\color[rgb]{0,0,0}$\ta{4}$}%
}}}}
\put(28396,-7891){\makebox(0,0)[lb]{\smash{{\SetFigFont{20}{24.0}{\familydefault}{\mddefault}{\updefault}{\color[rgb]{0,0,0}$\ta{7}$}%
}}}}
\put(36091,-7891){\makebox(0,0)[lb]{\smash{{\SetFigFont{20}{24.0}{\familydefault}{\mddefault}{\updefault}{\color[rgb]{0,0,0}$\ta{1}$}%
}}}}
\put(37981,-7891){\makebox(0,0)[lb]{\smash{{\SetFigFont{20}{24.0}{\familydefault}{\mddefault}{\updefault}{\color[rgb]{0,0,0}$\ta{6}$}%
}}}}
\put(32176,-7891){\makebox(0,0)[lb]{\smash{{\SetFigFont{20}{24.0}{\familydefault}{\mddefault}{\updefault}{\color[rgb]{0,0,0}$\ta{5}$}%
}}}}
\put(34336,-7891){\makebox(0,0)[lb]{\smash{{\SetFigFont{20}{24.0}{\familydefault}{\mddefault}{\updefault}{\color[rgb]{0,0,0}$\ta{2}$}%
}}}}
\put(30466,-8791){\makebox(0,0)[lb]{\smash{{\SetFigFont{20}{24.0}{\familydefault}{\mddefault}{\updefault}{\color[rgb]{0,0,0}$\ta{3}$}%
}}}}
\end{picture}%

%% file: 445c.pdf_t
\begin{picture}(0,0)%
\includegraphics{445c.pdf}%
\end{picture}%
\setlength{\unitlength}{4144sp}%
\begingroup\makeatletter\ifx\SetFigFont\undefined%
\gdef\SetFigFont#1#2#3#4#5{%
  \reset@font\fontsize{#1}{#2pt}%
  \fontfamily{#3}\fontseries{#4}\fontshape{#5}%
  \selectfont}%
\fi\endgroup%
\begin{picture}(48505,13053)(4666,-21041)
\put(11926,-20896){\makebox(0,0)[lb]{\smash{{\SetFigFont{12}{14.4}{\familydefault}{\mddefault}{\updefault}{\color[rgb]{0,0,0}$\Somea$}%
}}}}
\put(45271,-20896){\makebox(0,0)[lb]{\smash{{\SetFigFont{12}{14.4}{\familydefault}{\mddefault}{\updefault}{\color[rgb]{0,0,0}$\Somec$}%
}}}}
\put(4681,-12571){\makebox(0,0)[lb]{\smash{{\SetFigFont{12}{14.4}{\familydefault}{\mddefault}{\updefault}{\color[rgb]{0,0,0}$\ta{i_{t}}$}%
}}}}
\put(28441,-20986){\makebox(0,0)[lb]{\smash{{\SetFigFont{12}{14.4}{\familydefault}{\mddefault}{\updefault}{\color[rgb]{0,0,0}$\Someb$}%
}}}}
\put(38161,-12661){\makebox(0,0)[lb]{\smash{{\SetFigFont{12}{14.4}{\familydefault}{\mddefault}{\updefault}{\color[rgb]{0,0,0}$\ta{i_{t}}$}%
}}}}
\put(21376,-12661){\makebox(0,0)[lb]{\smash{{\SetFigFont{12}{14.4}{\familydefault}{\mddefault}{\updefault}{\color[rgb]{0,0,0}$\ta{i_{t}}$}%
}}}}
\put(31986,-11855){\makebox(0,0)[lb]{\smash{{\SetFigFont{12}{14.4}{\familydefault}{\mddefault}{\updefault}{\color[rgb]{0,0,0}$\ta{i_{s+1}}$}%
}}}}
\put(32106,-9838){\makebox(0,0)[lb]{\smash{{\SetFigFont{12}{14.4}{\familydefault}{\mddefault}{\updefault}{\color[rgb]{0,0,0}$\ta{i_{s+1}}$}%
}}}}
\put(9631,-14371){\makebox(0,0)[lb]{\smash{{\SetFigFont{12}{14.4}{\familydefault}{\mddefault}{\updefault}{\color[rgb]{0,0,0}$\ta{i_{s}}$}%
}}}}
\put(14311,-14416){\makebox(0,0)[lb]{\smash{{\SetFigFont{12}{14.4}{\familydefault}{\mddefault}{\updefault}{\color[rgb]{0,0,0}$\ta{i_{r}}$}%
}}}}
\put(19351,-14326){\makebox(0,0)[lb]{\smash{{\SetFigFont{12}{14.4}{\familydefault}{\mddefault}{\updefault}{\color[rgb]{0,0,0}$\ta{i_{r-1}}$}%
}}}}
\put(26326,-14461){\makebox(0,0)[lb]{\smash{{\SetFigFont{12}{14.4}{\familydefault}{\mddefault}{\updefault}{\color[rgb]{0,0,0}$\ta{i_{s}}$}%
}}}}
\put(31051,-14506){\makebox(0,0)[lb]{\smash{{\SetFigFont{12}{14.4}{\familydefault}{\mddefault}{\updefault}{\color[rgb]{0,0,0}$\ta{i_{r}}$}%
}}}}
\put(36046,-14416){\makebox(0,0)[lb]{\smash{{\SetFigFont{12}{14.4}{\familydefault}{\mddefault}{\updefault}{\color[rgb]{0,0,0}$\ta{i_{r-1}}$}%
}}}}
\put(43111,-14461){\makebox(0,0)[lb]{\smash{{\SetFigFont{12}{14.4}{\familydefault}{\mddefault}{\updefault}{\color[rgb]{0,0,0}$\ta{i_{s}}$}%
}}}}
\put(47791,-14506){\makebox(0,0)[lb]{\smash{{\SetFigFont{12}{14.4}{\familydefault}{\mddefault}{\updefault}{\color[rgb]{0,0,0}$\ta{i_{r}}$}%
}}}}
\put(52831,-14416){\makebox(0,0)[lb]{\smash{{\SetFigFont{12}{14.4}{\familydefault}{\mddefault}{\updefault}{\color[rgb]{0,0,0}$\ta{i_{r-1}}$}%
}}}}
\put(38071,-8836){\makebox(0,0)[lb]{\smash{{\SetFigFont{12}{14.4}{\familydefault}{\mddefault}{\updefault}{\color[rgb]{0,0,0}$\ta{i_{s+1}}$}%
}}}}
\end{picture}%

%% file: bibliography.bib
@inproceedings{ahpsv,
  author = {Oswin Aichholzer and Thomas Hackl and Alexander Pilz and Gelasio Salazar and Birgit Vogtenhuber},
  title = {{{Deciding monotonicity of good drawings of the complete graph}}},
  booktitle = {Proc. XVI Spanish Meeting on Computational Geometry (EGC 2015)},
  pages = {33--36},
  year = {2015},
  category = {3b},
  thackl_label = {47C},
  originalfile = {/geometry/cggg.bib}
}

@book{schaefer-2018-cng,
  title={Crossing Numbers of Graphs},
  author={Schaefer, Marcus},
  year={2018},
  publisher={CRC Press},
  doi={10.1201/9781315152394}
,
 note = {\doi{10.1201/9781315152394}}
}

@InProceedings{complexity_kyncl,
author="Kyn{\v{c}}l, Jan",
editor="Hong, Seok-Hee
and Nishizeki, Takao
and Quan, Wu",
title="The Complexity of Several Realizability Problems for Abstract Topological Graphs",
booktitle="Graph Drawing",
year="2008",
publisher="Springer Berlin Heidelberg",
address="Berlin, Heidelberg",
pages="137--158",
 note = {\doi{10.1007/978-3-540-77537-9_16}}
}

@article{extendingm,
  author       = {Kyn\v{c}l, Jan and Soukup, Jan},
  title        = {Extending simple monotone drawings},
 fjournal = {The Electronic Journal of Combinatorics},
 journal = {Electron. J. Comb.},
 issn = {1077-8926},
 volume = {33},
 number = {3},
 pages = {22 pages},
 note = {Id/No.~p3.4},
 year = {2026},
}

@article{howmany,
 author = {Pach, J{\'a}nos and T{\'o}th, G{\'e}za},
 title = {How many ways can one draw a graph?},
 fjournal = {Combinatorica},
 journal = {Combinatorica},
 issn = {0209-9683},
 volume = {26},
 number = {5},
 pages = {559--576},
 year = {2006},
 doi = {10.1007/s00493-006-0032-z},
 zbMATH = {5150174},
 Zbl = {1121.05006}
,
 note = {\doi{10.1007/s00493-006-0032-z}}
}

@article{enukyncl,
 author = {Kyn{\v{c}}l, Jan},
 title = {Enumeration of simple complete topological graphs},
 fjournal = {European Journal of Combinatorics},
 journal = {Eur. J. Comb.},
 issn = {0195-6698},
 volume = {30},
 number = {7},
 pages = {1676--1685},
 year = {2009},
 doi = {10.1016/j.ejc.2009.03.005},
 zbMATH = {5640346},
 Zbl = {1228.05175}
,
 note = {\doi{10.1016/j.ejc.2009.03.005}}
}

@article{gioan,
 author = {Gioan, Emeric},
 title = {Complete graph drawings up to triangle mutations},
 fjournal = {Discrete \& Computational Geometry},
 journal = {Discrete Comput. Geom.},
 issn = {0179-5376},
 volume = {67},
 number = {4},
 pages = {985--1022},
 year = {2022},
 doi = {10.1007/s00454-021-00339-8},
 zbMATH = {7526455},
 Zbl = {1489.05107}
,
 note = {\doi{10.1007/s00454-021-00339-8}}
}

@article{arroyogioan,
 author = {Arroyo, Alan and McQuillan, Dan and Richter, R. Bruce and Salazar, Gelasio},
 title = {Drawings of {{\(K_n\)}} with the same rotation scheme are the same up to {Reidemeister} moves ({Gioan}'s theorem)},
 fjournal = {The Australasian Journal of Combinatorics},
 journal = {Australas. J. Comb.},
 issn = {1034-4942},
 volume = {67},
 pages = {131--144},
 year = {2017},
 url = {ajc.maths.uq.edu.au/pdf/67/ajc_v67_p131.pdf},
 zbMATH = {6813586},
 Zbl = {1375.05180}
}

@article{schaeferdetour,
 author = {Schaefer, Marcus},
 title = {Taking a detour; or, {Gioan}'s theorem, and pseudolinear drawings of complete graphs},
 fjournal = {Discrete \& Computational Geometry},
 journal = {Discrete Comput. Geom.},
 volume = {66},
 number = {1},
 pages = {12--31},
 year = {2021},
 doi = {10.1007/s00454-021-00296-2},
 zbMATH = {7367362},
 Zbl = {1467.05178}
,
 note = {\doi{10.1007/s00454-021-00296-2}}
}

@incollection {fpss,
    AUTHOR = {Fulek, Radoslav and Pelsmajer, Michael J. and Schaefer, Marcus
              and {\v{S}}tefankovi{\v{c}}, Daniel},
     TITLE = {Hanani-{T}utte, monotone drawings, and level-planarity},
 BOOKTITLE = {Thirty essays on geometric graph theory},
     PAGES = {263--287},
 PUBLISHER = {Springer},
      YEAR = {2013},
       DOI = {10.1007/978-1-4614-0110-0_14},
       URL = {http://dx.doi.org/10.1007/978-1-4614-0110-0_14},
 note = {\doi{10.1007/978-1-4614-0110-0_14}}
}

@article{kyncl_monotone,
  author    = {Martin Balko and
               Radoslav Fulek and
               Jan Kyn{\v{c}}l},
  title     = {Crossing Numbers and Combinatorial Characterization of Monotone Drawings
               of {$K_n$}},
  journal   = {Discrete Comput. Geom.},
  volume    = {53},
  number    = {1},
  pages     = {107--143},
  year      = {2015},
  url       = {http://dx.doi.org/10.1007/s00454-014-9644-z},
  doi       = {10.1007/s00454-014-9644-z},
  bibsource = {dblp computer science bibliography, http://dblp.org},
 note = {\doi{10.1007/s00454-014-9644-z}}
}

@article{bergold2025plane,
  author    = {Bergold, Helena and Felsner, Stefan and Reddy, Meghana M. and Scheucher, Manfred},
  title     = {Plane {Hamiltonian} Cycles in Convex Drawings},
  journal   = {Discrete Comput. Geom.},
  year      = {2025},
  doi       = {10.1007/s00454-025-00752-3},
  url       = {https://doi.org/10.1007/s00454-025-00752-3},
 note = {\doi{10.1007/s00454-014-9644-z}}
}

@article{shellable_drawings,
  author    = {Bernardo M. {\'{A}}brego and
               Oswin Aichholzer and
               Silvia Fern{\'{a}}ndez{-}Merchant and
               Pedro Ramos and
               Gelasio Salazar},
  title     = {Shellable Drawings and the Cylindrical Crossing Number of {$K_n$}},
  journal   = {Discrete Comput. Geom.},
  volume    = {52},
  number    = {4},
  pages     = {743--753},
  year      = {2014},
  url       = {http://dx.doi.org/10.1007/s00454-014-9635-0},
  doi       = {10.1007/s00454-014-9635-0},
  bibsource = {dblp computer science bibliography, http://dblp.org}
,
 note = {\doi{10.1007/s00454-014-9635-0}}
}

@inproceedings{shootingstars,
  author       = {Oswin Aichholzer and
                  Alfredo Garc{\'{\i}}a and
                  Irene Parada and
                  Birgit Vogtenhuber and
                  Alexandra Weinberger},
  title        = {Shooting Stars in Simple Drawings of ${K}_{m,n}$},
  booktitle    = {{GD} 2022},
  series       = {LNCS},
  volume       = {13764},
  pages        = {49--57},
  publisher    = {Springer},
  year         = {2022},
  doi          = {10.1007/978-3-031-22203-0\_5}
,
 note = {\doi{10.1007/978-3-031-22203-0\_5}}
}

@inproceedings{fulekvargas,
  author       = {Radoslav Fulek and
                  Andres J. Ruiz{-}Vargas},
  title        = {Topological graphs: empty triangles and disjoint matchings},
  booktitle    = {SoCG 2013},
  pages        = {259--266},
  publisher    = {{ACM}},
  year         = {2013},
  doi          = {10.1145/2462356.2462394},
 note = {\doi{10.1145/2462356.2462394}}
}

@article{tangled-thrackle,
  author       = {Andres J. Ruiz{-}Vargas and
                  Andrew Suk and
                  Csaba D. T{\'{o}}th},
  title        = {Disjoint edges in topological graphs and the tangled-thrackle conjecture},
  journal      = {Eur. J. Comb.},
  volume       = {51},
  pages        = {398--406},
  year         = {2016},
  doi          = {10.1016/J.EJC.2015.07.004}
,
 note = {\doi{10.1016/J.EJC.2015.07.004}}
}

@article{twisted,
  author       = {Oswin Aichholzer and
                  Alfredo Garc{\'{\i}}a and
                  Javier Tejel and
                  Birgit Vogtenhuber and
                  Alexandra Weinberger},
  title        = {Twisted Ways to Find Plane Structures in Simple Drawings of Complete
                  Graphs},
  journal      = {Discret. Comput. Geom.},
  volume       = {71},
  number       = {1},
  pages        = {40--66},
  year         = {2024},
  doi          = {10.1007/S00454-023-00610-0},
 note = {\doi{10.1007/S00454-023-00610-0}}
}

@Article{kynclimproved,
 Author = {Kyn{\v{c}}l, Jan},
 Title = {Improved enumeration of simple topological graphs},
 FJournal = {Discrete \& Computational Geometry},
 Journal = {Discrete Comput. Geom.},
 ISSN = {0179-5376},
 Volume = {50},
 Number = {3},
 Pages = {727--770},
 Year = {2013},
 DOI = {10.1007/s00454-013-9535-8},
 zbMATH = {6223066},
 Zbl = {1275.05027}
,
 note = {\doi{10.1007/s00454-013-9535-8}}
}

@article {sukzeng,
    AUTHOR = {Suk, A. and Zeng, J.},
     TITLE = {Unavoidable Patterns in Complete Simple Topological Graphs},
   JOURNAL = {Discrete Comput. Geom.},
  FJOURNAL = {Discrete and Computational Geometry},
    VOLUME = {73},
      YEAR = {2025},
     PAGES = {79--91},
 note = {\doi{10.1007/s00454-024-00658-6}}
}

@article {pst,
    AUTHOR = {Pach, J\'anos and Solymosi, J\'ozsef and T\'oth, G\'eza},
     TITLE = {Unavoidable configurations in complete topological graphs},
   JOURNAL = {Discrete Comput. Geom.},
  FJOURNAL = {Discrete \& Computational Geometry. An International Journal
              of Mathematics and Computer Science},
    VOLUME = {30},
      YEAR = {2003},
    NUMBER = {2},
     PAGES = {311--320},
 note = {\doi{10.1007/s00454-003-0012-9}}
}
